\documentclass[12pt]{article}
\usepackage{cite}
\usepackage{CJK}
\usepackage{amssymb}
\usepackage{amsmath}
\usepackage{enumerate}
\usepackage{graphicx}
\usepackage{amsfonts}
\usepackage{url}
\usepackage{bm}
\usepackage{epstopdf}

\usepackage{pifont}
\usepackage{epsfig, subfigure, dsfont, amsthm, amsbsy, mathrsfs, amscd}
\newtheorem{theorem}{Theorem}

\newtheorem{corollary}{Corollary}

\newtheorem{lemma}{Lemma}

\def\be{\begin{equation}}
\def\ee{\end{equation}}
\def\ba{\begin{array}}
\def\ea{\end{array}}
\input amssym.def

\newcommand\btd{\raise 2pt \hbox{$\hat\bigtriangledown$}\hskip 1.5pt}
\newcommand\bt{\raise 2pt \hbox{$\bigtriangledown$}\hskip 1.5pt}
\usepackage{indentfirst}
\begin{document}

\title{Necessary conditions for classifying $\mathbf m$-separability of multipartite entanglements}

\author{Wen Xu$^{1}$,   Chuan-Jie Zhu$^{2,3}$,   Zhu-Jun Zheng$^{1}$  and Shao-Ming Fei$^{4,5}$ \\[10pt]
\small $^{1}$ Department of Mathematics, South China University of Technology,
Guangzhou 510640, P.R. China\\
\small $^{2}$ College of Mathematics and Physics Science, \\
\small Hunan University of Arts and Science, Changde 415000, P.R. China\\
\small $^{3}$ Department of Physics, Renmin University of China, Beijing 100872, P.R. China\\
\small $^{4}$ School of Mathematical Sciences, Capital Normal University, Beijing 100048, P.R.China\\
\small $^{5}$ Max-Planck-Institute for Mathematics in the Sciences, 04103 Leipzig, Germany}
\date{}
\maketitle

\begin{abstract}
We study the norms of the Bloch vectors for arbitrary $n$-partite quantum states. A tight upper bound of the norms is derived for $n$-partite systems with different individual dimensions. These upper bounds are used to deal with the separability problems. Necessary conditions are presented for $\mathbf m$-separable states in $n$-partite quantum systems. Based on the upper bounds, classification of multipartite entanglement is illustrated with detailed examples.

\end{abstract}

\section{Introduction}

Quantum entanglement is a remarkable resource in the theory of quantum information, with numerous applications in
quantum information processing, secure communication and channel protocols \cite{Eke,Ben,Bra}. A multipartite quantum state that is not separable with respect to any bipartition is said to be genuinely multipartite entangled \cite{Cir,Hub,Vic1}. Genuinely multipartite entangled states have significant advantages in quantum tasks compared with biseparable ones \cite{Hor}.

The notion of genuine multipartite entanglement (GME) was introduced in \cite{Hor}. Let $H_{i}^{d_{i}}, i=1,\cdots n,$ denote $d_{i}$-dimensional Hilbert spaces. An n-partite state $\rho\in H^{d_{1}}_{1}\otimes H^{d_{2}}_{2}\otimes\cdots\otimes H^{d_{n}}_{n}$ can be expressed as $\rho=\sum\limits_{i}p_{i}|\psi_{i}\rangle\langle \psi_{i}|,$ where $\sum\limits_{i}p_{i}=1,  0< p_{i}\leq1, |\psi_{i}\rangle\in H^{d_{1}}_{1}\otimes H^{d_{2}}_{2}\otimes\cdots\otimes H^{d_{n}}_{n}$ are normalized pure states. $\rho$ is biseparable if $|\psi_{i}\rangle~(i=1,\cdots,n)$ can be expressed as one of the forms: $|\psi_{i}\rangle=|\psi_{i}^{j_{1}\cdots j_{k-1}}\rangle\otimes|\psi_{i}^{j_{k}\cdots j_{n}}\rangle,$ where $|\psi_{i}^{j_{1}\cdots j_{k-1}}\rangle$ and $|\psi_{i}^{j_{k}\cdots j_{n}}\rangle$ denote pure states in $H^{j_{1}}_{i}\otimes\cdots \otimes H^{j_{k-1}}_{i}$ and $H^{j_{k}}_{i}\otimes\cdots \otimes H^{j_{n}}_{i}$, respectively, $j_{1}\neq\cdots \neq j_{n}\in\{1, \cdots, n\}$, $k=2,...,n-1$.
Otherwise, $\rho$ is called genuine multipartite entangled. Correspondingly, we say that the state
$\rho=\sum\limits_{i}p_{i}|\psi_{i}\rangle\langle \psi_{i}|$ is $\mathbf m$-separable
if all the $|\psi_{i}\rangle$ are tensor products of m vectors in the subspaces of $H^{d_{1}}_{1}\otimes H^{d_{2}}_{2}\otimes\cdots\otimes H^{d_{n}}_{n}$.

Any quantum state has Bloch representation in multipartite high-dimensional quantum systems. By using the norms of the Bloch vectors, the density operators in lower dimensions were discussed in \cite{Jak,Kim}. For bipartite and multipartite quantum states, separable conditions have been presented in \cite{Vic2,Vic3,Has,Wang}. The norms of the Bloch vectors for any qudit quantum states with subsystems less than or equal to four have been investigated in \cite{Li}. Then in \cite{Pop}, T\u{a}n\u{a}sescu et al. generalized the result of \cite{Li} for
four-partite quantum systems, which provided an upper bound on the entanglement measure given by the Bloch vector norm and a necessary algebraic condition for separability of a general multi-partite quantum system under any arbitrary partition function. Two multipartite entanglement measures for n-qubit and n-qudit pure states are given in \cite{Joa1,Joa2}. In \cite{Yu}, the sum of relative isotropic strengths of any three-qudit state over $d$-dimensional Hilbert space cannot exceed $d-1$ have been discussed, and the trade-off relations and monogamy-like relations of the
sum of spin correlation strengths for pure three- and four-partite systems are derived. Some sufficient or necessary conditions of GME were presented in \cite{van,Zhang,Shen}. To the detection of GME, the common criterion is the entanglement witnesses \cite{Hub,Ban,Jun,Wu}. In \cite{Vic1}, the norms of the Bloch vectors give rise to a general framework to detect different classes of GME for arbitrary dimensional quantum systems. Recently, based on the norms of the correlation tensors, Zhao et al. \cite{Zhao} have been studied the separability criteria by matrix method and necessary conditions of separability for multipartite systems are given under arbitrary partition.

In this paper, we study the Bloch representations of quantum states with arbitrary number of subsystems.
In Section 2, we present tight upper bounds for the norms of Bloch vectors in n-qudit quantum states. These upper bounds are then used to derive tight upper bounds for entanglement measures in \cite{Joa1,Joa2}. The upper bounds of the norms of the Bloch vectors are useful to study the separability. In Section 3, we investigate different subclasses of bi-separable states in n-partite systems. Necessary conditions for $\mathbf m$-separability and complete classification of n-partite quantum systems are presented.

\section{Upper bounds of the norms of Bloch vectors}

Let $\lambda_{i}$, $i=1, \cdots, d^{2}-1,$ denote the generators of the special unitary group $SU(d)$, which satisfy $\lambda_{i}^{\dag}=\lambda_{i}$, $Tr(\lambda_{i})=0$, $Tr(\lambda_{i}\lambda_{j})=2\delta_{ij}$.
The following theorem gives the general result for n-partite quantum states.

\begin{theorem}
Let $\rho\in H^{d_{1}}_{1}\otimes H^{d_{2}}_{2}\otimes\cdots\otimes H^{d_{n}}_{n}$ $(n\geq3, 2\leq d_{1}\leq d_{2}\leq \cdots\leq d_{n}, d_{n}\leq d_{1}\cdots d_{n-1})$ be an n-partite quantum state. We have
\begin{equation}
\left\|\mathbf T^{(12\cdots n)}\right\|^{2}\leq 2^{n}\left\{1-\frac{\sum\limits_{1\leq i_{1}<\cdots<i_{n-1}\leq n}d_{i_{1} }\cdots d_{i_{n-1}}-\sum\limits_{i=1}^{n}d_{i}+(n-2)d_{n}}{(n-2)d_{1}\cdots d_{n-1}d_{n}^{2}} \right\}.
\end{equation}
\end{theorem}

\begin{proof}
$\rho$ has the Bloch representation:
\begin{equation}
\begin{split}
\rho=&\frac{1}{d_{1}\cdots d_{n}}I_{d_{1}}\otimes\cdots\otimes I_{d_{n}}+\frac{1}{2}\bigg(\frac {1}{d_{2}
\cdots d_{n}}\sum_{i_{1}=1}^{d_{1}^{2}-1}t_{i_{1}}^{(1)}\lambda_{i_{1}}\otimes I_{d_{2}}\otimes\cdots\otimes I_{d_{n}} + \cdots \\
&+ \frac {1}{d_{1}\cdots d_{n-1}}\sum_{i_{n}=1}^{d_{n}^{2}-1}t_{i_{n}}^{(n)}I_{d_{1}} \otimes I_{d_{2}}\otimes\cdots\otimes \lambda_{i_{n}}\bigg)+\cdots\\
&+ \frac{1}{2^{n}}\sum_{k=1}^{n}\sum_{i_{k}=1}^{d_{k}^{2}-1} t_{i_{1}\cdots i_{n}}^{(1\cdots n)}\lambda_{i_{1}}\otimes\lambda_{i_{2}}\otimes \cdots \otimes \lambda_{i_{n}},
\end{split}
\end{equation}
where $I_{d_{i}}$ denotes the $d_{i}\times d_{i}$ identity matrix $,i=1,\cdots,n$,
$t_{i_{1}}^{(1)}= Tr(\rho\lambda_{i_{1}}\otimes I_{d_{2}}\otimes \cdots \otimes I_{d_{n}}), \cdots,
t^{(j_{1}\cdots j_{k})}_{i_{j_{1}}\cdots i_{j_{k}}}=Tr(\rho\lambda_{i_{j_{1}}}\otimes\cdots\otimes\lambda_{i_{j_{k}}} \otimes I_{d_{j_{k+1}}}\otimes\cdots \otimes I_{d_{j_{n}}}),\cdots, t_{i_{1}\cdots i_{n}}^{(1\cdots n)}=Tr(\rho\lambda_{i_{1}}\otimes\lambda_{i_{2}} \otimes \cdots \otimes \lambda_{i_{n}})$ and $\mathbf T^{(1)}, \cdots,\mathbf T^{(j_{1}\cdots j_{k})},\cdots, \mathbf T^{(1\cdots n)}$ are the vectors~(tensors)~with the elements $t_{i_{1}}^{(1)}, \cdots, t_{i_{j_{1}}\cdots i_{j_{k}}}^{(j_{1}\cdots j_{k})},\cdots,t_{i_{1}\cdots i_{n}}^{(1\cdots n)}~(1\leq j_{1}<\cdots<j_{k}\leq n,i_{s}=1, \cdots, d_{s}^{2}-1, s=1,\cdots,n)$, respectively.

Set
$$
\begin{array}{c}
\left\|\mathbf T^{(1)}\right\|^{2}= \sum\limits_{i_{1}=1}^{d_{1}^{2}-1}\left(t_{i_{1}}^{(1)}\right)^{2}, \\
\cdots,\\
\left\|\mathbf T^{(j_{1}\cdots j_{k})}\right\|^{2}=\sum\limits_{s=j_{1}}^{j_{k}}\sum\limits_{i_{s}=1}^{d_{s}^{2}-1}\left(t_{i_{j_{1}}\cdots i_{j_{k}}}^{(j_{1}\cdots j_{k})}\right)^{2},\\
\cdots,\\
\left\|\mathbf T^{(1\cdots n)}\right\|^{2}=\sum\limits_{s=1}^{n}\sum\limits_{i_{s}=1}^{d_{s}^{2}-1}\left(t_{i_{1}\cdots i_{n}}^{(1\cdots n)}\right)^{2},\\
\\
x_{1}=\frac{1}{d_{2}\cdots d_{n}}\left\|\mathbf T^{(1)}\right\|^{2}+\cdots+\frac{1}{d_{1}\cdots d_{n-1}}\left\|\mathbf T^{(n)}\right\|^{2}, \\
\\
x_{2}=\frac{1}{d_{3}\cdots d_{n}}\left\|\mathbf T^{(12)}\right\|^{2}+\cdots+\frac{1}{d_{1}\cdots d_{n-2}}\left\|\mathbf T^{(n-1, n)}\right\|^{2}, \\
\cdots,\\
x_{n}=\left\|\mathbf T^{(12\cdots n)}\right\|^{2}.
\end{array}
$$
For a pure state $\rho=|\psi\rangle\langle \psi|$, one has $Tr(\rho^{2})=1$, namely,
\begin{equation}\label{3}
Tr(\rho^{2})=\frac{1}{d_{1}\cdots d_{n}}+\frac{1}{2}x_{1}+\frac{1}{2^{2}}x_{2}+\cdots+\frac{1}{2^{n}}x_{n}=1.
\end{equation}
In the following we denote $\rho_{j_{1}}, \rho_{j_{2}\cdots j_{n}}$ the reduced density matrix for the subsystem $H_{j_1}^{d_{j_1}}$ and $H_{j_2}^{d_{j_2}}\otimes\cdots \otimes H_{j_n}^{d_{j_n}}$, $j_{1}\neq \cdots\neq j_{n}\in\{1, 2, \cdots, n\}$. One computes that,

\begin{equation*}
\rho_{j_{1}}=\frac{1}{d_{j_{1}}}I_{d_{j_{1}}}+
\frac{1}{2}\sum\limits_{i_{j_{1}}=1}^{d_{j_{1}}^{2}-1}t_{i_{j_{1}}}^{(j_{1})}\lambda_{i_{j_{1}}},
\end{equation*}
\begin{equation*}
\begin{split}
\rho_{j_{2}\cdots j_{n}}=&\frac{1}{d_{j_{2}}\cdots d_{j_{n}}}I_{d_{j_{2}}}\otimes \cdots \otimes I_{d_{j_{n}}}+\\
&\frac{1}{2}\bigg(\frac{1}{d_{j_{3}}\cdots d_{j_{n}}}\sum\limits_{i_{j_{3}}=1}^{d_{j_{3}}^{2}-1}t_{i_{j_{3}}}^{(j_{3})}\lambda _{i_{j_{3}}}\otimes \cdots \otimes I_{d_{j_{n}}}+\cdots\\
&+\frac{1}{d_{j_{2}}\cdots d_{j_{n-1}}}\sum\limits_{i_{j_{n}}=1}^{d_{j_{n}}^{2}-1}t_{i_{j_{n}}}^{(j_{n})}I_{d_{j_{2}}} \otimes \cdots \otimes \lambda_{i_{j_{n}}}\bigg) +\cdots\\
&+
\frac{1}{2^{n-1}}\sum_{s=2}^{n}\sum_{i_{j_{s}}=1}^{d_{j_{s}}^{2}-1} t_{i_{j_{2}}\cdots i_{j_{n}}}^{(j_{2}\cdots j_{n})}\lambda_{i_{j_{2}}}\otimes \cdots \otimes \lambda_{i_{j_{n}}}.
\end{split}
\end{equation*}
For a pure state $\rho=|\psi\rangle\langle \psi|$, we have
\begin{equation}
Tr(\rho_{j_{1}}^{2})=Tr(\rho_{j_{2}\cdots j_{n}}^{2}),
\end{equation}
which holds for any $j_{1}\neq\cdots \neq j_{n}\in{\{1, 2, \cdots, n\}}$. We obtain
\begin{equation}
\sum_{j_{1}=1}^{n}\frac{1}{d_{j_{1}}}Tr(\rho_{j_{1}}^{2})=\sum_{j_{1}=1}^{n}\frac{1}{d_{j_{1}}}
Tr(\rho_{j_{2}\cdots j_{n}}^{2}).
\end{equation}
Hence we get
\begin{equation*}
\begin{split}
\sum_{i=1}^{n}\frac{1}{d_{i}^{2}}+\frac{1}{2}\sum_{i=1}^{n}\frac{1}{d_{i}}\left\|\mathbf T^{(i)}\right\|^{2}=&\frac{n}{d_{1}\cdots d_{n}}+\frac{n-1}{2}x_{1}+\\
&\frac{n-2}{2^{2}}x_{2}+\cdots+\frac{1}{2^{n-1}}x_{n-1}.
\end{split}
\end{equation*}
Then
\begin{equation}\label{6}
\begin{split}
\frac{1}{2^{2}}x_{2}=&\frac{1}{n-2}\left(\sum_{i=1}^{n}\frac{1}{d_{i}^{2}}-\frac{n}{d_{1}\cdots d_{n}}\right)+\\
&\frac{1}{2(n-2)}\left(\sum_{i=1}^{n}\frac{1}{d_{i}}\left\|\mathbf T^{(i)}\right\|^{2}-
(n-1)x_{1}\right)\\
&-\cdots-\frac{1}{2^{n-1}(n-2)}x_{n-1}.
\end{split}
\end{equation}
Substituting (\ref{6}) into (\ref{3}), we get
\begin{equation*}
\begin{split}
\frac{1}{2^{n}}x_{n}=&\left[1-\frac{1}{n-2}\left(\sum_{i=1}^{n}\frac{1}{d_{i}^{2}}-\frac{2}{d_{1}\cdots d_{n}}\right)\right]-\\
&\frac{1}{2(n-2)}\left(\sum_{i=1}^{n}\frac{1}{d_{i}}\left\|\mathbf T^{(i)}\right\|^{2}-x_{1}\right)-\cdots-\frac{n-3}{2^{n-1}(n-2)}x_{n-1}\\
\leq &1-\frac{\sum\limits_{1\leq i_{1}<\cdots<i_{n-1}\leq n}d_{i_{1} }\cdots d_{i_{n-1}}-\sum\limits_{i=1}^{n}d_{i}+(n-2)d_{n}}{(n-2)d_{1}\cdots d_{n-1}d_{n}^{2}}
\end{split}
\end{equation*}
where the inequality holds for
\begin{equation}\label{7}
\sum\limits_{i=1}^{n}\frac{1}{d_{i}^{2}}-\frac{2}{d_{1}\cdots d_{n}}\geq
\frac{\sum\limits_{1\leq i_{1}<\cdots<i_{n-1}\leq n}d_{i_{1} }\cdots d_{i_{n-1}}-\sum\limits_{i=1}^{n}d_{i}+(n-2)d_{n}}{d_{1}\cdots d_{n-1}d_{n}^{2}}
\end{equation}
\begin{equation}\label{8}
\sum\limits_{i=1}^{n}\frac{1}{d_{i}}\left\|\mathbf T^{(i)}\right\|^{2}-x_{1}=\sum\limits_{i=1}^{n}\left(\frac{1}{d_{i}}-\frac{1}{d_{1}\cdots d_{i-1}d_{i+1}\cdots d_{n}}\right)\left\|\mathbf T^{(i)}\right\|^{2}\geq0,
\end{equation}
the inequation (\ref{8}) holds for $d_{i}\leq d_{1}\cdots d_{i-1}d_{i+1}\cdots d_{n}, i=1,\cdots,n$.
And $x_{3},\cdots,x_{n-1}\geq0$. In fact, the inequation (\ref{7}) holds if and only if
\begin{equation}\label{9}
\sum\limits_{i=1}^{n-1}(d_{n}-d_{i})(\frac{1}{d_{i}^{2}}-\frac{1}{\prod\limits_{i=1}^{n}d_{i}})
\geq0.
\end{equation}
Since $d_{n}\geq d_{i}$, and $d_{i}^{2}\leq\prod\limits_{i=1}^{n}d_{i},i=1,\cdots,n-1$, so the inequation (\ref{9}) holds, which is equivalent to hold for inequation (\ref{7}). Hence we get
\begin{equation*}
 x_{n}\leq2^{n}\left\{1-\frac{\sum\limits_{1\leq i_{1}<\cdots<i_{n-1}\leq n}d_{i_{1} }\cdots d_{i_{n-1}}-\sum\limits_{i=1}^{n}d_{i}+(n-2)d_{n}}{(n-2)d_{1}\cdots d_{n-1}d_{n}^{2}} \right\},~~~n\geq 3.
\end{equation*}

Then we consider a mixed state $\rho$ with ensemble representation
$\rho=\sum\limits_{i}p_{i}|\psi_{i}\rangle\langle \psi_{i}|$, where\\$\sum\limits_{i}p_{i}=1, 0<p_{i}\leq1$, by the convexity of the Frobenius norm one derives
\begin{equation*}
\begin{split}
\left\|\mathbf T^{(12\cdots n)}(\rho)\right\|^{2}=&\left\|\sum\limits_{i}p_{i}\mathbf T^{(12\cdots n)}(|\psi_{i}\rangle\langle \psi_{i}|)\right\|^{2}\\
\leq&\sum\limits_{i}p_{i}\left\|\mathbf T^{(12\cdots n)}(|\psi_{i}\rangle\langle \psi_{i}|)\right\|^{2}\\
\leq&2^{n}\left\{1-\frac{\sum\limits_{1\leq i_{1}<\cdots<i_{n-1}\leq n}d_{i_{1} }\cdots d_{i_{n-1}}-\sum\limits_{i=1}^{n}d_{i}+(n-2)d_{n}}{(n-2)d_{1}\cdots d_{n-1}d_{n}^{2}} \right\},~~~n\geq 3,
\end{split}
\end{equation*}
which ends the proof.
\end{proof}

{\sf Remark 1:} Theorem 1 is a generalization of Proposition 1 and Proposition 2 given in \cite{Zhao}.  When $n=3, 2\leq d_{1}\leq d_{2}\leq d_{3}, d_{3}\leq d_{1}d_{2}$, we obtain that
\begin{equation*}
\left\|\mathbf T^{(123)}\right\|^{2}\leq 8\left(1-\frac{d_{1}d_{2}+d_{1}d_{3}+d_{2}d_{3}-d_{1}-d_{2}}{d_{1}d_{2}d_{3}^{2}}\right)
\end{equation*}
which coincide with the upper bound in \cite{Zhao}. When $n=4, 2\leq d_{1}\leq d_{2}\leq d_{3}\leq d_{4}, d_{4}\leq d_{1}d_{2}d_{3}$, we obtain that
\begin{equation*}
\left\|\mathbf T^{(1234)}\right\|^{2}\leq 16\left(1-\frac{d_{1}d_{2}d_{3}+d_{1}d_{2}d_{4}+d_{1}d_{3}d_{4}+d_{2}d_{3}d_{4}-d_{1}-d_{2}-d_{3}+d_{4}}{d_{1}d_{2}d_{3}d_{4}^{2}}\right)
\end{equation*}
which also coincide with the upper bound in \cite{Zhao}.

As a special case, consider $ d_{1}=\cdots=d_{n}=d$ in Theorem 1. We have

\begin{corollary}
Let $\rho\in H^{d}_{1}\otimes H^{d}_{2}\otimes\cdots\otimes H^{d}_{n}$ $(n\geq3, d\geq2)$ be an n-qudit quantum state. We have
\begin{equation}
\left\|\mathbf T^{(12\cdots n)}\right\|^{2}\leq \frac{2^{n}\left[(n-2)d^{n}-nd^{n-2}+2\right]}{(n-2)d^{n}}.
\end{equation}
\end{corollary}

{\sf Remark 2:} The upper bound of Corollary 1 is the same as in \cite{Pop} and \cite{Zhao}. And Corollary 1 is the generalization of the results of \cite{Li}. When $n=3,4$, the results of Corollary 1 reduce to the ones in \cite{Li} and \cite{Pop}. And when $n=3$, the upper bound of Corollary 1 is $\frac{8(d-1)^{2}(d+2)}{d^{3}}$, which tighter than the upper bound of Corollary 2.2 given in \cite{Yu}.

The Bloch vectors are used to define a valid entanglement measure in \cite{Joa1,Joa2} as follows. For an n-qudit pure state, the entanglement measure is defined as:
\begin{equation}
E_{\mathbf T}(|\psi\rangle)=\left(\frac{d}{2}\right)^{\frac{n}{2}}\left\|\mathbf T^{(1\cdots n)}\right\|-\left(\frac{d(d-1)}{2}\right)^{\frac{n}{2}},
\end{equation}
where $\mathbf T^{(1\cdots n)}$ is defined as a vector with elements $t_{i_{1}\cdots i_{n}}^{(1\cdots n)}=Tr(\rho\lambda_{i_{1}}\otimes\lambda_{i_{2}} \otimes \cdots \otimes \lambda_{i_{n}})$.
Our results can give rise to an upper bound of the entanglement:

\begin{corollary}
For any n-qudit pure state $\rho\in H^{d}_{1}\otimes H^{d}_{2}\otimes\cdots\otimes H^{d}_{n}$ $(n\geq3, d\geq2)$, the entanglement measure has the upper bounds:
\begin{equation}
E_{\mathbf T}(|\psi\rangle)\leq\left(\frac{d}{2}\right)^{\frac{n}{2}}
\left[\sqrt{d^{n}-\frac{n}{n-2}d^{n-2}+\frac{2}{n-2}}-(d-1)^{\frac{n}{2}}\right],
\end{equation}
which coincide with the upper bound in \cite{Pop}.
\end{corollary}

\section{The Necessary conditions for $\mathbf m$-separable states}

Now we study separability problems of n-partite quantum systems based on the upper bounds of the norms of Bloch vectors. Let $\rho\in H^{d_{1}}_{1}\otimes H^{d_{2}}_{2}\otimes\cdots\otimes H^{d_{n}}_{n}$, $n\geq3$, $2\leq d_{1}\leq d_{2}\leq \cdots\leq d_{n}$. If $\rho$ can be written as $\rho=\sum\limits_{i}p_{i}|\psi_{i}\rangle\langle \psi_{i}|$, where $\sum\limits_{i}p_{i}=1$, $0< p_{i}\leq 1$, $|\psi_{i}\rangle$ is one of the following sets: $\left\{|\phi_{j_{1}}\rangle\otimes|\phi_{j_{2}\cdots j_{n}}\rangle\right\}, \cdots, \left\{|\phi_{1}\rangle\otimes\cdots\otimes|\phi_{n}\rangle\right\}, \left\{|\phi_{j_{1}j_{2}}\rangle\otimes|\phi_{j_{3}\cdots j_{n}}\rangle\right\}, \cdots,$ where $j_{1}\neq\cdots\neq j_{n}\in\{1, \cdots n\}.$ Then $\rho$ is called $(1,n-1)$ separable, $\cdots, \underbrace{(1,\cdots,1)}_{n}$ separable, $(2,n-2)$ separable, $\cdots$, respectively.

\begin{lemma}
Let $\rho_{j}\in H_{j}^{d_{j}}$ $(j=1,\cdots, n, d_{j}\geq2)$ be the reduced density operator of $\rho$. We have
\begin{equation}\label{l1}
\left\|\mathbf T^{(j)}\right\|^{2}\leq\frac{2(d_{j}-1)}{d_{j}}.
\end{equation}
\end{lemma}

\begin{proof}
$\rho_{j}$ has the Bloch representation:
\begin{equation}
\rho_{j}=\frac{1}{d_{j}}I_{d_{j}}+\frac{1}{2}\sum\limits_{i_{j}=1}^{d_{j}^{2}-1}t_{i_{j}}^{(j)}\lambda_{i_{j}},
\end{equation}
where $t_{i_{j}}^{(j)}= Tr(\rho_{j}\lambda_{i_{j}})$, $\mathbf T^{(j)}$ is a vector with entries $t_{i_{j}}^{(j)}$, $j=1,\cdots, n$, $i_{j}=1,\cdots,d_{j}^{2}-1$.
Since $Tr(\rho_{j}^{2})\leq1$, i.e. $\frac{1}{d_{j}}+\frac{1}{2}\left\|\mathbf T^{(j)}\right\|^{2}\leq1$, one obtains (\ref{l1}).
\end{proof}

\begin{lemma}
Let $\rho_{jk}\in H_{j}^{d_{j}}\otimes H_{k}^{d_{k}}$ $(1\leq j<k\leq n, 2\leq d_{j}\leq d_{k})$ be the reduced density operator of $\rho$. We have
\begin{equation}
\left\|\mathbf T^{(jk)}\right\|^{2}\leq\frac{2^{2}(d_{j}^{2}-1)}{d_{j}^{2}}.
\end{equation}
\end{lemma}

\begin{proof}
$\rho_{jk}$ has the Bloch representation:
\begin{equation}
\begin{split}
\rho_{jk}=&\frac{1}{d_{j}d_{k}}I_{d_{j}}\otimes I_{d_{k}}+
\frac{1}{2d_{k}}\sum\limits_{i_{j}=1}^{d_{j}^{2}-1}t_{i_{j}}^{(j)}\lambda_{i_{j}}\otimes I_{d_{k}}+\\
&\frac{1}{2d_{j}}\sum\limits_{i_{k}=1}^{d_{k}^{2}-1}t_{i_{k}}^{(k)}I_{d_{j}}\otimes\lambda_{i_{k}}
+\frac{1}{4}\sum\limits_{i_{j}=1}^{d_{j}^{2}-1}\sum\limits_{i_{k}=1}^{d_{k}^{2}-1}t_{i_{j}i_{k}}^{(jk)}
\lambda_{i_{j}}\otimes\lambda_{i_{k}},
\end{split}
\end{equation}
where $t_{i_{j}}^{(j)}= Tr(\rho_{jk}\lambda_{i_{j}}\otimes I_{d_{k}})$, $t_{i_{k}}^{(k)}= Tr(\rho_{jk} I_{d_{j}}\otimes\lambda_{i_{k}})$,
$t_{i_{j}i_{k}}^{(jk)}= Tr(\rho_{jk}\lambda_{i_{j}}\otimes\lambda_{i_{k}})$. $\mathbf T^{(j)}$, $\mathbf T^{(k)}$ and $\mathbf T^{(jk)}$ are vectors with entries $t_{i_{j}}^{(j)},~t_{i_{k}}^{(k)}$ and $t_{i_{j}i_{k}}^{(jk)}$ $(1\leq j<k\leq n, i_{j}=1,\cdots,d_{j}^{2}-1, i_{k}=1,\cdots,d_{k}^{2}-1)$.
Set
$$
\begin{array}{cc}
\left\|\mathbf T^{(j)}\right\|^{2}=\sum\limits_{i_{j}=1}^{d_{j}^{2}-1}\left(t_{i_{j}}^{(j)}\right)^{2}, \\[2mm]
\left\|\mathbf T^{(k)}\right\|^{2}=\sum\limits_{i_{k}=1}^{d_{k}^{2}-1}\left(t_{i_{k}}^{(k)}\right)^{2}, \\[2mm]
\left\|\mathbf T^{(jk)}\right\|^{2}=\sum\limits_{i_{j}=1}^{d_{j}^{2}-1}\sum\limits_{i_{k}=1}^{d_{k}^{2}-1}\left(t_{i_{j}i_{k}}^{(jk)}\right)^{2}.
\end{array}
$$
For a pure state $\rho_{jk}=|\psi\rangle\langle \psi|$, one has $Tr(\rho_{jk}^{2})=1$, namely,
\begin{equation}
Tr(\rho_{jk}^{2})=\frac{1}{d_{j}d_{k}}+\frac{1}{2d_{k}}\left\|\mathbf T^{(j)}\right\|^{2}+ \frac{1}{2d_{j}}\left\|\mathbf T^{(k)}\right\|^{2}+\frac{1}{4}\left\|\mathbf T^{(jk)}\right\|^{2}=1.
\end{equation}
Let $\rho_{j}$ and $\rho_{k}$ be the reduced density matrices with respect to the subsystems $1\leq j<k\leq n$. Since for a pure state $\rho_{jk}$, $Tr(\rho_{j}^{2})=Tr(\rho_{k}^{2})$, i.e. $\frac{1}{d_{j}}+\frac{1}{2}\left\|\mathbf T^{(j)}\right\|^{2}=\frac{1}{d_{k}}+\frac{1}{2}\left\|\mathbf T^{(k)}\right\|^{2}$. Therefore, we get
\begin{equation}
\begin{split}
\left\|\mathbf T^{(jk)}\right\|^{2}=& \frac{2^{2}(d_{j}^{2}-1)}{d_{j}^{2}}-\frac{2(d_{j}+d_{k})}{d_{j}d_{k}}
\left\|\mathbf T^{(j)}\right\|^{2}\\
\leq&\frac{2^{2}(d_{j}^{2}-1)}{d_{j}^{2}}.
\end{split}
\end{equation}

Now we consider a mixed state $\rho_{jk}$ with ensemble representation
$\rho_{jk}=\sum\limits_{i}p_{i}|\psi_{i}\rangle\langle \psi_{i}|$, where
$\sum\limits_{i}p_{i}=1$, $0<p_{i}\leq1$, by the convexity of the Frobenius norm one derives
\begin{equation*}
\begin{split}
\left\|\mathbf T^{(jk)}(\rho_{jk})\right\|^{2}=&\left\|\sum\limits_{i}p_{i}\mathbf T^{(jk)}(|\psi_{i}\rangle\langle \psi_{i}|)\right\|^{2}\\
\leq&\sum\limits_{i}p_{i}\left\|\mathbf T^{(jk)}(|\psi_{i}\rangle\langle \psi_{i}|)\right\|^{2}\\
\leq&\frac{2^{2}(d_{j}^{2}-1)}{d_{j}^{2}},
\end{split}
\end{equation*}
which ends the proof.
\end{proof}

For the $n$-partite quantum state $\rho\in H^{d_{1}}_{1}\otimes H^{d_{2}}_{2}\otimes\cdots\otimes H^{d_{n}}_{n},n\geq3,2\leq d_{1}\leq\cdots \leq d_{n}$, consider the $\mathbf m$-partition of $n$-qudit quantum state $\rho$, we denote $\mathbf m=(k_{1},\cdots,k_{m})$, where $\sum\limits_{s=1}^{m}k_{s}=n,1\leq k_{1}\leq \cdots\leq k_{m}\leq n-1$ and $n_{j}=\sum\limits_{s=1}^{j}k_{s},1\leq j\leq m, 1\leq n_{j}\leq n$. By Theorem 1, $d_{n_{j}}\leq d_{n_{j-1}+1}~d_{n_{j-1}+2}\cdots d_{n_{j}-1}$ if $3\leq k_{j}\leq n-1, 1\leq j\leq m$. Moreover, denote the following:
\begin{equation}\label{21}
\begin{array}{cc}
a_{1}=\frac{2\left(d_{1}-1\right)}{d_{1}},\\
\cdots,\\
a_{n_{p}}=\frac{2\left(d_{n_{p}}-1\right)}{d_{n_{p}}},\\
\\
a_{n_{b}}=\frac{2^{2}\left(d_{n_{b}-1}^{2}-1\right)}{d_{n_{b}-1}^{2}},\\
\cdots,\\
a_{n_{b+(q-1)}}=\frac{2^{2}\left(d_{n_{b+(q-1)}-1}^{2}-1\right)}{d_{n_{b+(q-1)}-1}^{2}}\\
\cdots,\\
a_{n_{c}}=2^{t}\left\{1-\frac{\sum\limits d_{i_{n_{c}-(t-1)}}\cdots d_{i_{n_{c}-1}}-\sum\limits_{i=0}^{t-1}d_{n_{c}+i}~+(n-2)d_{n_{c}}}{(t-2)d_{n_{c}-(t-1)}\cdots d_{n_{c}-1}d_{n_{c}}^{2} }\right\},\\
\cdots,\\
a_{n_{c+(s-1)}}=2^{t}\left\{1-\frac{\sum\limits d_{i_{n_{c+(s-1)}-(t-1)}}\cdots d_{i_{n_{c+(s-1)}-1}}-\sum\limits_{i=0}^{t-1}d_{n_{c}+(s-1)+i}+(n-2)d_{n_{c+(s-1)}}}{(t-2)d_{n_{c+(s-1)}-(t-1)}\cdots d_{n_{c+(s-1)}-1}d_{n_{c+(s-1)}}^{2} }\right\},~~~t\geq3.
\end{array}
\end{equation}
where $p+q+\cdots+s=m,~p+2q+\cdots+ts=n,1\leq p,b,q,\cdots,c,s\leq m$.

Then from Lemmas 1, 2 and Theorem 1, we have
\begin{equation*}
\begin{array}{cc}
 \left\|\mathbf T^{(1)}\right\|^{2}\leq a_{1},\\
 \cdots,\\
 \left\|\mathbf T^{(n_{p})}\right\|^{2}\leq a_{n_{p}},\\
 \\
 \left\|\mathbf T^{\left(n_{b}-1,n_{b}\right)}\right\|^{2}\leq a_{n_{b}},\\
 \cdots,\\
 \left\|\mathbf T^{\left(n_{b+(q-1)}-1,n_{b+(q-1)}\right)}\right\|^{2}\leq a_{n_{b+(q-1)}},\\
 \\
 \left\|\mathbf T^{\left(n_{c}-(t-1),\cdots,n_{c}\right)}\right\|^{2}\leq a_{n_{c}},\\
 \cdots,\\
 \left\|\mathbf T^{\left(n_{c+(s-1)}-(t-1),\cdots,n_{c+(s-1)}\right)}\right\|^{2}\leq a_{n_{c+(s-1)}},~~~t\geq3.
\end{array}
\end{equation*}
where $\mathbf T^{(u)}$ is a vector with the entries $t_{i_{u}}^{(u)}~(u=1,\cdots,n_{p}, \;i_{u}=1,\cdots,d_{u}^{2}-1).$
$\mathbf T^{(xy)}$ is a vector with entries $t_{i_{x}i_{y}}^{(xy)}~ (x=n_{b}-1,\cdots,n_{b+(q-1)}-1,
y=n_{b},\cdots,n_{b+(q-1)}, i_{x}=1,\cdots,d_{x}^{2}-1, i_{y}=1,\cdots,d_{y}^{2}-1).$ $\cdots$. $\mathbf T^{(x_{1},\cdots,x_{t})}$ is a vector with entries $t_{i_{x_{1}}\cdots i_{x_{t}}}^{(x_{1}\cdots x_{t})} ~(x_{1}=n_{c}-(t-1),\cdots,n_{c+(s-1)}-(t-1), \cdots, x_{t}=n_{c},\cdots,n_{c+(s-1)}, i_{x_{1}}=1,\cdots,d_{x_{1}}^{2}-1,\cdots,
i_{x_{t}}=1,\cdots,d_{x_{t}}^{2}-1).$

The following theorem gives the necessary conditions of $\mathbf m$-separability.
\begin{theorem}
Let $\rho\in H^{d_{1}}_{1}\otimes H^{d_{2}}_{2}\otimes\cdots\otimes H^{d_{n}}_{n}$ $(n\geq3, 2\leq d_{1}\leq d_{2}\leq \cdots\leq d_{n})$ be an n-partite quantum state. If $\rho$ is $\mathbf m$-separable we have
\begin{equation}\label{th2}
\left\|\mathbf T^{(12\cdots n)}\right\|^{2}\leq\prod\limits_{f=1}^{n_{p}}a_{f}
\prod\limits_{g=n_{b}}^{n_{b+(q-1)}}a_{g}\cdots\prod\limits_{h=n_{c}}^{n_{c+(s-1)}}a_{h},
\end{equation}
where $\mathbf m=(k_{1},\cdots,k_{m})$, $\sum\limits_{s=1}^{m}k_{s}=n,1\leq k_{1}\leq \cdots\leq k_{m}\leq n-1$ and $n_{j}=\sum\limits_{s=1}^{j}k_{s},1\leq j\leq m, 1\leq n_{j}\leq n$, $d_{n_{j}}\leq d_{n_{j-1}+1}~d_{n_{j-1}+2}\cdots d_{n_{j}-1}$ if $3\leq k_{j}\leq n-1, 1\leq j\leq m$. $a_{f},a_{g},a_{h} ~(f=1, \cdots, n_{p},g=n_{b},\cdots,n_{b+(q-1)},h=n_{c},\cdots,n_{c+(s-1)})$ are given in (\ref{21}).
\end{theorem}

\begin{proof}
If $\rho=|\psi\rangle\langle \psi|$ is an $\mathbf m$-separable pure state, where $|\psi\rangle\in H^{d_{1}}_{1}\otimes H^{d_{2}}_{2}\otimes\cdots\otimes H^{d_{n}}_{n}.$ Without lose of generality, assume that
\begin{equation*}
\begin{split}
|\psi\rangle=&|\phi_{1}\rangle\otimes\cdots\otimes|\phi_{n_{p}}\rangle\otimes\left|\phi_{n_{b}-1,n_{b}}\right\rangle\otimes\cdots
\otimes\left|\phi_{n_{b+(q-1)}-1,n_{b+(q-1)}}\right\rangle\otimes\cdots\\
&\otimes\left|\phi_{n_{c}-(t-1),\cdots,n_{c}}\right\rangle\otimes\cdots\otimes
\left|\phi_{n_{c+(s-1)}-(t-1),\cdots,n_{c+(s-1)}}\right\rangle.
\end{split}
\end{equation*}
We have
\begin{equation}
\begin{split}
t_{i_{1}\cdots i_{n}}^{\left(1\cdots n\right)}=&Tr\left(|\psi\rangle\langle \psi|\lambda_{i_{1}}\otimes\cdots\otimes\lambda_{i_{n}}\right)\\
=&Tr\left(|\phi_{1}\rangle\langle\phi_{1}|\lambda_{i_{1}}\right)\cdots
Tr\left(|\phi_{n_{p}}\rangle\langle\phi_{n_{p}}|\lambda_{i_{n_{p}}}\right)\cdot
\\
&Tr\left(\left|\phi_{n_{b}-1,n_{b}}\right\rangle\left\langle\phi_{n_{b}-1,n_{b}}\right|\lambda_{i_{n_{b}-1}}\otimes
\lambda_{i_{n_{b}}}\right)\cdots\\
&Tr\bigg(\left|\phi_{n_{b+(q-1)}-1,n_{b+(q-1)}}\right\rangle\left\langle\phi_{n_{b+(q-1)}-1,n_{b+(q-1)}}\right|\\
&\lambda_{i_{n_{b+(q-1)}-1,n_{b+(q-1)}}}\otimes
\lambda_{i_{n_{b+(q-1)}-1,n_{b+(q-1)}}}\bigg)\cdots\\
&Tr\bigg(\left|\phi_{n_{c}-(t-1),\cdots,n_{c}}\right\rangle\left\langle\phi_{n_{c}-(t-1),\cdots,n_{c}}\right|\\
&\lambda_{i_{n_{c}-(t-1)}}\otimes\cdots\otimes\lambda_{i_{n_{c}}}\bigg)\cdots
\\
&Tr\bigg(\left|\phi_{n_{c+(s-1)}-(t-1),\cdots,n_{c+(s-1)}}\right\rangle
\left\langle\phi_{n_{c+(s-1)}-(t-1),\cdots,n_{c+(s-1)}}\right|\\
&\lambda_{i_{n_{c+(s-1)}-(t-1)}}\otimes\cdots\otimes\lambda_{i_{n_{c+(s-1)}}}\bigg)\\
=&t_{i_{1}}^{\left(1\right)}\cdots t_{i_{n_{p}}}^{\left(n_{p}\right)}~t_{i_{n_{b}-1}i_{n_{b}}}^{\left(n_{b}-1,n_{b}\right)}\cdots t_{i_{n_{b+(q-1)}-1}i_{n_{b+(q-1)}}}^{\left(n_{b+(q-1)}-1,n_{b+(q-1)}\right)}\cdots\\
&t_{i_{n_{c}-(t-1)}\cdots i_{n_{c}}}^{\left(n_{c}-(t-1),\cdots,n_{c}\right)}\cdots t_{i_{n_{c+(s-1)}-(t-1)}\cdots i_{n_{c+(s-1)}}}^{\left(n_{c+(s-1)}-(t-1),\cdots,n_{c+(s-1)}\right)}.
\end{split}
\end{equation}
Thus
\begin{equation}
\begin{split}
\left\|\mathbf T^{(12\cdots n)}\right\|^{2}=&\left\|\mathbf T^{(1)}\right\|^{2}\cdots\left\|\mathbf T^{(n_{p})}\right\|^{2}\\
&\left\|\mathbf T^{(n_{b}-1,n_{b})}\right\|^{2}\cdots
\left\|\mathbf T^{(n_{b+(q-1)}-1,n_{b+(q-1)})}\right\|^{2}\cdots\\
&\left\|\mathbf T^{(n_{c}-(t-1),\cdots,n_{c})}\right\|^{2}\cdots\left\|\mathbf T^{(n_{c+(s-1)}-(t-1),\cdots,n_{c+(s-1)})}\right\|^{2}\\
\leq&\prod\limits_{f=1}^{n_{p}}a_{f}
\prod\limits_{g=n_{b}}^{n_{b+(q-1)}}a_{g}\cdots\prod\limits_{h=n_{c}}^{n_{c+(s-1)}}a_{h}.
\end{split}
\end{equation}

Then for any mixed state $\rho=\sum\limits_{k}p_{k}|\psi_{k}\rangle\langle \psi_{k}|\in H^{d_{1}}_{1}\otimes H^{d_{2}}_{2}\otimes\cdots\otimes H^{d_{n}}_{n}$, where $\sum\limits_{k}p_{k}=1$,
$0< p_{k}\leq1$,by the convexity of the Frobenius norm one derives
\begin{equation}
\begin{split}
\left\|\mathbf T^{(12\cdots n)}(\rho)\right\|^{2}=&\left\|\sum\limits_{k}p_{k}\mathbf T^{(12\cdots n)}(|\psi_{k}\rangle\langle \psi_{k}|)\right\|^{2}\\
\leq&\sum\limits_{k}p_{k}\left\|\mathbf T^{(12\cdots n)}(|\psi_{k}\rangle\langle \psi_{k}|)\right\|^{2}\\
\leq&\prod\limits_{f=1}^{n_{p}}a_{f}
\prod\limits_{g=n_{b}}^{n_{b+(q-1)}}a_{g}\cdots\prod\limits_{h=n_{c}}^{n_{c+(s-1)}}a_{h}.
\end{split}
\end{equation}
\end{proof}

{\sf Remark 3:} The upper bounds of Theorem 2 is a generalization of Theorem 2 given in \cite{Pop} and Theorem 7 given in \cite{Zhao}, respectively. Set $d_{1}=\cdots=d_{n}=d$, and $a_{1}=\cdots=a_{n_{p}}$, $a_{n_{b}}=\cdots=a_{n_{b+(q-1)}}=a_{2}$, $\cdots$, $a_{n_{c}}=\cdots=a_{n_{c+(s-1)}}=a_{t}$. Then
(\ref{th2}) gives rise to
\begin{equation}
\left\|\mathbf T^{(12\cdots n)}\right\|^{2}\leq a_{1}^{p}a_{2}^{q}\cdots a_{t}^{s}.
\end{equation}
which coincide with Theorem 2 in \cite{Pop} and Theorem 7 in \cite{Zhao}.

{\sf Remark 4:} Let $\rho \in H_{1}^{d_{1}}\otimes H_{2}^{d_{2}}\otimes H_{3}^{d_{3}}\otimes H_{4}^{d_{4}}$ be a four-partite quantum state. One has
\begin{equation}
\left\|\mathbf T^{(1234)}\right\|^{2}\leq
\begin{cases}
\frac{2^{4}(d_{1}-1)}{d_{1}}\left\{1-\frac{d_{2}d_{3}+d_{2}d_{4}+d_{3}d_{4}-2d_{4}}{d_{2}d_{3}d_{4}^{2}}\right\}, &
if\;\rho\;is\;(1,3)\;separable;\\
\\
\frac{2^{4}\left(d_{1}^{2}-1\right)\left(d_{3}^{2}-1\right)}{d_{1}^{2}d_{3}^{2}}, &if\;\rho\;is\;(2,2)\;separable;\\
\\
\frac{2^{4}\left(d_{1}-1\right)\left(d_{2}-1\right)\left(d_{3}^{2}-1\right)}{d_{1}d_{2}d_{3}^{2}}, &if\;\rho\;is\;(1,1,2)\;separable;\\
\\
\frac{2^{4}\prod\limits_{i=1}^{4}\left(d_{i}-1\right)}{\prod\limits_{i=1}^{4}d_{i}}, &if\;\rho\;is\;(1,1,1,1)\;separable.
\end{cases}
\end{equation}

The following two examples show that the upper bounds in Theorem 2 are nontrivial and are tight.

{\it Example 1:} Consider the quantum state $\rho \in H_{1}^{2}\otimes H_{2}^{2}\otimes H_{3}^{2}\otimes H_{4}^{2}\otimes H_{5}^{2}$,
\begin{equation}
\rho =x(|\psi\rangle\langle \psi|+|\varphi\rangle\langle \varphi|)+\frac{1-2x}{32}I_{32},
\end{equation}
where $|\psi\rangle=\frac{1}{\sqrt{2}}(|00000\rangle+|11111\rangle)$, $|\varphi\rangle=\frac{1}{2}(|00001\rangle+
|00010\rangle+|00100\rangle+|01000\rangle)$. Since $\left\|\mathbf T^{(12345)}\right\|^{2}=\sum\limits_{i_{1}, \cdots, i_{5}=1}^{3}\left(t_{i_{1}\cdots i_{5}}^{(1\cdots5)}\right)^{2}$, where $t_{i_{1}\cdots i_{5}}^{(1\cdots5)}=Tr(\rho\lambda_{i_{1}}\otimes\cdots\otimes\lambda_{i_{5}})$ are the entries of $\mathbf T^{(12345)}$, we have $\left\|\mathbf T^{(12345)}\right\|^{2}=20x^{2}.$ Thus for $\frac{3\sqrt{5}}{10}<x\leq\frac{\sqrt{15}}{5},$
$\rho$ is not $(1,4)$ or $(1,2,2)$ separable. For $\frac{\sqrt{15}}{5}<x\leq1,$ $\rho$ is not $(2,3)$
separable. For $\frac{\sqrt{5}}{5}<x\leq\frac{3\sqrt{5}}{10},$ $\rho$ is not $(1,1,3)$ separable. For
$\frac{\sqrt{15}}{10}<x\leq\frac{\sqrt{5}}{5},$ $\rho$ is not $(1,1,1,2)$ separable. For
$\frac{\sqrt{5}}{10}<x\leq\frac{\sqrt{15}}{10},$ $\rho$ is not $(1,1,1,1,1)$ separable.

{\it Example 2:} Consider the quantum state $\rho \in H_{1}^{2}\otimes H_{2}^{3}\otimes H_{3}^{4}\otimes H_{4}^{5}$,
\begin{equation}
\rho =x|\psi\rangle\langle \psi|+\frac{1-x}{120}I_{120},
\end{equation}
where $|\psi\rangle=\frac{1}{\sqrt{2}}(|0\rangle_{1}|0\rangle_{2}|0\rangle_{3}|4\rangle_{4}+
|1\rangle_{1}|0\rangle_{2}|0\rangle_{3}|0\rangle_{4})$, $|0\rangle_{1}:=[1,0]^{T}$, $|1\rangle_{1}:=[0,1]^{T}$,
$|0\rangle_{2}:=[1,0,0]^{T}$, $|0\rangle_{3}:=[1,0,0,0]^{T}$, $|0\rangle_{4}:=[1,0,0,0,0]^{T}$, $|4\rangle_{4}:=[0,0,0,0,1]^{T}$ (T is the transpose). Since $\left\|\mathbf T^{(1234)}\right\|^{2}=\sum\limits_{k=1}^{4}\sum\limits_{i_{k}=1}^{d_{k}^{2}-1}
\left(t_{i_{1}\cdots i_{4}}^{(1\cdots4)}\right)^{2}$, where $t_{i_{1}\cdots i_{4}}^{(1\cdots4)}=Tr(\rho\lambda_{i_{1}}\otimes\cdots\otimes\lambda_{i_{4}})$ are the entries of $\mathbf T^{(1234)}$, we can compute that $\left\|\mathbf T^{(1234)}\right\|^{2}=6x^{2}.$
Thus for $\frac{\sqrt{263}}{15}<x\leq\frac{\sqrt{30}}{4},$ $\rho$ is not $(1,3)$ separable. For $\frac{\sqrt{30}}{4}<x\leq1,$ $\rho$ is not $(2,2)$ separable. For $\frac{\sqrt{30}}{6}<x\leq\frac{\sqrt{263}}{15},$ $\rho$ is not $(1,1,2)$ separable. For $\frac{2\sqrt{30}}{15}<x\leq\frac{\sqrt{30}}{6},$ $\rho$ is not $(1,1,1,1)$ separable.

From the above results, we are able to classify the entanglement of n-partite quantum states by using the norms of the Bloch vector $\left\|\mathbf T^{(12\cdots n)}\right\|^{2}$. The upper bounds of $\left\|\mathbf T^{(12\cdots n)}\right\|^{2}$ can be used to identify the $\mathbf m$-separable n-partite quantum states, which include the fully separable states and the genuine multipartite entangled states as
special classes.

\section{Conclusion}
Classification and detection of quantum entanglement are basic and fundamental problems in theory of quantum entanglement. We have investigated the norms of the Bloch vectors for arbitrary n-partite quantum systems. Tight upper bounds of the norms have been derived, and used to derive tight upper bounds for entanglement measure defined by the norms of Bloch vectors. The upper bounds have a close relationship to the separability. Necessary conditions have been presented for $\mathbf m$-separable quantum states.
With these upper bounds a complete classification of n-partite quantum states has been obtained.
Our results may highlight further studies on the quantum entanglement.

\section*{ Acknowledgments} This work is supported by the NSF of China under
Grant Nos. 11571119 and 11675113, the Key Project of Beijing Municipal Commission of Education (Grant No.
KZ201810028042), and Beijing Natural Science Foundation (Z190005).
It is a pleasure to thank Jin-Wei Huang for helpful discussion.

\end{document}